\documentclass[aps,pra,superscriptaddress,12pt,longbibliography,notitlepage,onecolumn]{revtex4-2}

\makeatletter
\renewcommand{\p@subsection}{}
\renewcommand{\p@subsubsection}{}
\makeatother

\usepackage{amsmath,amsfonts,amssymb,amsthm}
\usepackage{mathtools}
\usepackage{setspace}
\usepackage[normalem]{ulem} 
\usepackage{physics}
\usepackage{verbatim}
\usepackage{xspace}

\usepackage[nolists,heads,nomarkers]{endfloat}

\usepackage{graphicx,xcolor}
\usepackage{url}

\usepackage{verbatim}
\usepackage{alltt}
\usepackage{moreverb}
\usepackage{bm}
\usepackage{bbold}

\usepackage{xr}

\usepackage{caption}
\usepackage{subcaption}

\providecommand{\ignore}[1]{}

\newif\ifcmnt
\cmnttrue
\ifdefined\cmntsoff\cmntfalse\fi

\ifcmnt

\else

\fi

\newcommand{\cB}{\mathcal{B}}

\newcommand{\cN}{\mathcal{N}}

\numberwithin{equation}{section}

\newtheorem*{theorem*}{Theorem}

\newtheorem*{lemma*}{Lemma}

\newcommand{\id}{\text{id}}

\usepackage{hyperref}
\hypersetup{pdfstartview=}

\usepackage{tikz}
\usetikzlibrary{cd}
\makeatletter
\let\tikzcd@@original\tikzcd@%
\let\endtikzcd@original\endtikzcd%
\makeatother
\usetikzlibrary{quantikz}

\usepackage{enumitem}

\numberwithin{equation}{section}
\newtheorem{thm}{Theorem}[section]
\newtheorem{prop}[thm]{Proposition}

\newtheorem{lem}[thm]{Lemma}

\theoremstyle{definition}

\newcommand{\A}{\mathsf{A}}
\newcommand{\B}{\mathsf{B}}
\newcommand{\C}{\mathsf{C}}

\newcommand{\cvbell}{CV Bell}

\newcommand{\pamp}{a}
\newcommand{\pnoise}{n}

\begin{document}

\title{Constraints on Gaussian Error Channels and Measurements for Quantum Communication}
\author{Alex Kwiatkowski}
\affiliation{National Institute of Standards and Technology, Boulder, Colorado 80305, USA}
\affiliation{Department of Physics, University of Colorado, Boulder, Colorado, 80309, USA}
\author{Ezad Shojaee}
\affiliation{National Institute of Standards and Technology, Boulder, Colorado 80305, USA}
\affiliation{Department of Physics, University of Colorado, Boulder, Colorado, 80309, USA}
\author{Sristy Agrawal}
\affiliation{National Institute of Standards and Technology, Boulder, Colorado 80305, USA}
\affiliation{Department of Physics, University of Colorado, Boulder, Colorado, 80309, USA}
\author{Akira Kyle}
\affiliation{National Institute of Standards and Technology, Boulder, Colorado 80305, USA}
\affiliation{Department of Physics, University of Colorado, Boulder, Colorado, 80309, USA}
\author{Curtis Rau}
\affiliation{National Institute of Standards and Technology, Boulder, Colorado 80305, USA}
\affiliation{Department of Physics, University of Colorado, Boulder, Colorado, 80309, USA}
\author{Scott Glancy}
\affiliation{National Institute of Standards and Technology, Boulder, Colorado 80305, USA}
\author{Emanuel Knill}
\affiliation{National Institute of Standards and Technology, Boulder, Colorado 80305, USA}
\affiliation{Center for Theory of Quantum Matter, University of Colorado, Boulder, Colorado 80309, USA}
\begin{abstract}
Joint Gaussian measurements of two quantum systems can be used for quantum communication between remote parties, as in teleportation or entanglement swapping protocols.
Many types of physical error sources throughout a protocol can be modeled by independent Gaussian error channels acting prior to measurement.
In this work we study joint Gaussian measurements on two modes $\A$ and $\B$ that take place after independent single-mode Gaussian error channels, for example loss with parameters $l_\A$ and $l_\B$ followed by added noise with parameters $\pnoise_\A$ and $\pnoise_\B$.
We show that, for any Gaussian measurement, if $l_\A + l_\B + \pnoise_\A + \pnoise_\B \geq 1$ then the effective total measurement is separable and unsuitable for teleportation or entanglement swapping of arbitrary input states. If this inequality is not satisfied then there exists a Gaussian measurement that remains inseparable. We extend the results and determine the set of pairs of single-mode Gaussian error channels that render all Gaussian measurements separable.
\end{abstract}
\maketitle
\section{Introduction}
Quantum communication between remote parties is a key requirement in the establishment of a quantum network~\cite{Kozlowski_2020}.
One important tool for quantum communication is a joint measurement of two parties, as in teleportation~\cite{PhysRevLett.70.1895,PhysRevLett.80.869,Braunstein_2005}  or entanglement swapping~\cite{PhysRevA.61.010302}.
Photonic modes are typically the easiest quantum systems to directly transmit between remote parties, and this motivates a joint measurement of two bosonic systems~\cite{Weedbrook_2012,Adesso_2014,WANG_2007} as a natural choice for quantum communication.
This choice allows the remote parties to in principle transfer entanglement to quantum systems that can be discrete-variable, non-optical, or otherwise physically different from the transmitted quantum systems~\cite{PhysRevLett.97.150403,PhysRevLett.114.100501,https://doi.org/10.1002/que2.60,PhysRevLett.83.2095}.
For example, entanglement swapping can be combined with transduction~\cite{Lauk_2020} to establish entanglement between remote systems that cannot be directly transmitted.
Although both Gaussian ~\cite{PhysRevA.89.022331,PhysRevA.60.2752,dellanno2016nongaussian,Tian:18}
and non-Gaussian~\cite{Lim_2016,PhysRevA.94.032333} joint measurements have been studied in this context and realized experimentally~\cite{Furusawa706,PhysRevLett.117.240503,PhysRevLett.94.220502,PhysRevLett.93.250503,PhysRevLett.114.100501,Takeda2013-kv},
high-efficiency Gaussian measurements are more readily available.
Furthermore, many types of physical error sources throughout a protocol can be modeled by independent Gaussian error channels acting prior to measurement~\cite{Holevo2007,cerf_leuchs_polzik,PhysRevA.50.3295,Noh_2020}.
The combination of an ideal Gaussian measurement with Gaussian error channels results in an effective total measurement that is also Gaussian~\cite{holevo2020structure}. If this effective Gaussian measurement is separable, meaning that the POVM elements of the measurement are all convex combinations of positive product operators, then it is unsuitable for teleportation or entanglement swapping regardless of the input states.

Motivated by applications to quantum communication, we investigate independent pairs of Gaussian error channels for which all effective Gaussian measurements are separable.
Special cases of this scenario have been studied previously, usually in the context of entanglement swapping with a joint Gaussian measurement on specific input states. In these studies, the joint Gaussian measurement is often the continuous-variable (CV) Bell measurement, which can be implemented by opposite-quadrature measurements after a balanced beamsplitter. In particular, Ref.~\cite{Hoelscher_Obermaier_2011} analyzes entanglement swapping in the specific case where two parties each prepare the same two-mode Gaussian state and each send one mode through lossy channels to a swapping station that implements a CV Bell measurement.
Implicit in that analysis is the requirement $l_\A + l_\B < 1$ for establishing remote entanglement with that scheme, where
$l_{\A}$ and $l_{\B}$ are the loss parameters of the two channels. (We provide definitions of parameters for loss, noise, and other channels in Section~\ref{section:preliminaries}.)
In a similar scenario, Ref.~\cite{Pirandola_2021} analyzes entanglement swapping in the case where the two parties hold identical two-mode-squeezed-vacuum states subjected to loss $l$ followed by symmetric displacement noise that can be classically correlated.
In the special case of independent noise channels with parameter $\pnoise$, the results of that analysis imply that $l+\pnoise<\tfrac{1}{2}$ is necessary to produce remote entanglement after the CV Bell measurement at the swapping station.
Similarly, Ref.~\cite{PhysRevLett.97.150403} analyzes entanglement swapping with a CV Bell measurement when both parties start with the same two-mode Gaussian state and send one mode to the swapping station and implicitly finds the requirement $l + \pnoise  < \tfrac{1}{2}$ for establishing remote entanglement.
Ref.~\cite{PhysRevA.89.022331} analyzes a modified entanglement swapping scenario in which each party sends two modes of a local tripartite quantum state to a central station, where a \cvbell{} measurement is performed on two modes and homodyne measurements are made on the other two modes. For this scenario Ref.~\cite{PhysRevA.89.022331} provides a necessary and sufficient condition for successful swapping of entanglement for Gaussian input states in terms of the purities of the reduced input states.

In addition to all-Gaussian scenarios, entanglement swapping to establish entanglement between discrete-variable systems that are initially entangled with CV modes is studied in Refs.~\cite{parker2017hybrid,Parker_2020}. In this case, the CV modes to be measured experience low levels of asymmetric loss before they are combined on a balanced beam-splitter and one mode is projected into vacuum while a homodyne measurement is made on the other~\cite{parker2017hybrid,Parker_2020}.
This swapping measurement is Gaussian and our analysis applies.
Similarly, Ref.~\cite{Silveri_2016} analyzes a procedure to entangle remote transmon qubits by first locally entangling them with CV modes and then performing an effective \cvbell{} measurement on the CV modes. This analysis also arrives at the requirement $l < \tfrac{1}{2}$ for symmetric loss $l$ in order to generate any entanglement.

In this work we analyze amplification channels with parameters $\pamp_\A$ and $\pamp_\B$ followed by loss channels with parameters $l_\A$ and $l_\B$ that act on two modes $\mathsf{A}$ and $\mathsf{B}$ prior to measurement and prove that, for an arbitrary Gaussian measurement $G$, if $l_\A + l_\B \geq 1$ then the effective measurement is separable, regardless of $\pamp_\A$ and $\pamp_\B$ (Prop.~\ref{prop:amp_loss}).
Entanglement swapping is possible with a \cvbell{} measurement when
$l_\A + l_\B < 1$ and for any $\pamp_\A, \pamp_\B$, so the inequality $l_\A + l_\B \geq 1$ for separability of all
  effective measurements is tight (Prop.~\ref{prop:amp_loss_tight}). We call this inequality
  the loss condition for separability. We demonstrate
that the loss condition for separability can alternatively be obtained by applying the dual
error channels to the POVM elements of an arbitrary ideal Gaussian measurement
and showing that they become separable when this condition is met (Section~\ref{section:dual_channels}). Furthermore, we determine the set of pairs of single-mode Gaussian error channels that have the property that all effective Gaussian measurements are separable (Section~\ref{section:all_pairs}). We achieve this by using the classification of Gaussian channels in
Ref.~\cite{Holevo2007}
to reduce every non-trivial pair to the case of channels consisting of amplification followed by loss.
For example, if the channels consist
of loss $l_{\A}$ and $l_{B}$ followed by added noise $\pnoise_{\A}$ and
$\pnoise_{B}$, then all effective Gaussian measurements are separable when
$l_\A + l_\B + \pnoise_\A + \pnoise_\B \geq 1$ (Prop.~\ref{prop:loss_noise}).

This paper is organized as follows. In Section~\ref{section:preliminaries} we provide preliminaries about Gaussian states, measurements, and channels.
In Section~\ref{section:loss_threshold} we prove Prop.~\ref{prop:loss_threshold}, which directly establishes the loss condition for separability in the special case that the amplification parameters are zero.
In Section~\ref{section:sm_channels} we establish the loss condition for separability for arbitrary amplification parameters in Prop.~\ref{prop:amp_loss} and obtain the corresponding condition when the channels are loss followed by noise in Prop.~\ref{prop:loss_noise}.
In Section~\ref{section:all_pairs} we provide a method of determining, for any pair of single-mode Gaussian error channels that act prior to measurement, whether or not all effective Gaussian measurements are separable.
In Section~\ref{section:dual_channels} we show how the loss condition for separability can be obtained by applying the dual error channels to the POVM elements of a Gaussian measurement.

\section{Preliminaries}
\label{section:preliminaries}
In this paper we assume familiarity with continuous-variable quantum mechanics and Gaussian quantum information as described, for example, in the review article Ref.~\cite{Weedbrook_2012}. We also assume a basic understanding of Gaussian channels as described in Ref.~\cite{Holevo2007}.

A bosonic system of $m$ modes is characterized by $m$ creation and annihilation operators $\{\hat{a}_i, \hat{a}^\dagger_i\}_{i=1}^m$ with commutation relations $[\hat{a}_i, \hat{a}_j] = 0$ and $[\hat{a}_i, \hat{a}^\dagger_j] = \delta_{ij}$.
We work with quadrature operators $\hat{x}_i = (\hat{a}_i + \hat{a}_i^\dagger)/\sqrt{2}$ and $\hat{p}_i = i(\hat{a}_i^\dagger - \hat{a}_i)/\sqrt{2}$ with commutation relations $[\hat{x}_i, \hat{p}_j] = i\delta_{ij}$.
We label invidudual modes as $\mathsf{A,B,C,D,E,F}$, and the annihilation operators corresponding to those modes are denoted by $\hat{a}, \hat{b}, \hat{c},\hat{d},\hat{e},\hat{f}$ respectively.
We also occasionally use $\A$ and $\B$ to refer to multi-mode systems, in which case the annihilation operators are denoted $\hat{a}_i$ or $\hat{b}_i$ where $i$ is the mode index.

The single-mode Gaussian channels that are especially relevant are loss, amplification, and added noise, which are referred to in Ref.~\cite{Ivan_2011} as types $\mathcal{C}_1, \mathcal{C}_2$, and $\mathcal{B}_2$ respectively, and are phase-insensitive.
We choose the conventions that the covariance matrix of vacuum is $\mathbb{1}/2$, a loss channel with parameter $0 \leq l \leq 1$ transforms covariance matrices $V$
according to $V \mapsto (1-l)V + l \cdot\mathbb{1}/2$, an amplification channel with parameter $\pamp > 1$ acts according to $V \mapsto \pamp V + (\pamp-1)\cdot\mathbb{1}/2$, and a noise channel with parameter $\pnoise$ acts according to $V \mapsto V + \pnoise\cdot\mathbb{1}$.
A particular family of single-mode channels that we study in this paper is realized by amplification with parameter $\pamp$ followed by loss with parameter $l$, which transforms covariance matrices $V$ according to

\begin{equation}
\label{eq:channel_cov_update}
V \mapsto \pamp(1-l)V + (\pamp(1-l) + 2l - 1)\cdot \mathbb{1}/2.
\end{equation}
We emphasize that a Gaussian channel is uniquely defined, up to displacement, by its action on covariance matrices.

Some Gaussian channels are entanglement-breaking, meaning that for any input state the output system is unentangled from every other system. We refer to Ref.~\cite{holevo2008entanglementbreaking} for further information about entanglement-breaking Gaussian channels.

 Following Refs.~\cite{PhysRevA.66.032316,holevo2020structure} we adopt the definition of a noiseless Gaussian measurement on a system $\mathsf{A}$ of $m$ modes to be any measurement consisting of an arbitrary Gaussian unitary followed by homodyne measurements on $k$ modes and heterodyne measurements on the remaining $m-k$ modes. We use the definitions of homodyne and heterodyne measurement in Ref.~\cite{holevo2020structure}, and we summarize them here.
A homodyne measurement is a von-Neumann measurement of a quadrature $\hat{x}$ with a POVM that is the spectral measure of the operator $\hat{x}$, and has the associated resolution of the identity conventionally written as $\int dx \ketbra{x}$.
A heterodyne measurement is defined by the POVM given by the resolution of the identity $\text{id} = \frac{1}{\pi}\int_\mathbb{C} d\alpha \ketbra{\alpha}$, where $\ketbra{\alpha}$ is the projector onto the coherent state labeled by $\alpha$~\cite{Adesso_2014}.
We refer to the improper projectors $\ketbra{x}$
and the projectors $\ketbra{\alpha}$ as the POVM elements of the homodyne and the heterodyne measurements, respectively. We use the notation $\ket{x=0}$ to denote the improper $0$-eigenstate of the operator $\hat{x}$ and $\ket{0}$ to denote vacuum, which is the $0$-eigenstate of the operator $\hat{a}$.
The results of Ref.~\cite{holevo2020structure} imply that a noiseless Gaussian measurement of a joint system $\A\B$ where system $\B$ is initialized in a pure Gaussian state is equivalent to a noiseless Gaussian measurement of system $\A$ alone. It follows that our results for arbitrary Gaussian measurements hold even when Gaussian ancillas are permitted.
Furthermore, Ref.~\cite{holevo2020structure} implies that a noiseless Gaussian measurement of a system of $m$ modes can be characterized by a family $\{M_{i}\}_{i=1}^{m}$ of linear combinations
of creation and annihilation operators
with the properties in List~\ref{list:gaussian_meas_props}.
\newcounter{g_meas_props}
\begin{equation}
\label{list:gaussian_meas_props}
\begin{aligned}
\refstepcounter{g_meas_props}(\roman{g_meas_props})\quad & \{M_{i}\}_{i=1}^{m} \;\text{are mutually commuting}\\
\refstepcounter{g_meas_props}(\roman{g_meas_props})\quad & \{M_{i}\}_{i=1}^{m} \;\text{are linearly independent}\\
\refstepcounter{g_meas_props}(\roman{g_meas_props})\quad & M_i = M_i^\dagger, i = 1\dots k, \;\text{for some integer} \;k\\
\refstepcounter{g_meas_props}(\roman{g_meas_props})\quad & \left[M_j, M_{j^\prime}^\dagger\right]=\delta_{j j^\prime}, \; k+1 \leq j,j^\prime \leq m.
\end{aligned}
\end{equation}
Thus, the first $k$ operators behave like
quadrature operators, and the last $m-k$ behave like annihilation operators.
We note that $k$ can be $0$, in which case all of the operators $\{M_i\}_{i=1}^m$ satisfy property $(iv)$ and none of them satisfy property $(iii)$.
If the operators $\{M_i\}_{i=1}^m$ satisfy the properties in List~\ref{list:gaussian_meas_props}, then they can be mapped to the operators $\hat{x}_1,...,\hat{x}_k, \hat{a}_{k+1},...,\hat{a}_m$ by a Gaussian unitary $U$.
This follows from the fact that the $m$ operators
$\{M_1,...,M_k,(M_{k+1} + M_{k+1}^\dagger)/\sqrt{2},...,(M_{m} + M_{m}^\dagger)/\sqrt{2}\}$ and the $m-k$ operators $\{i(M_{k+1}^\dagger - M_{k+1})/\sqrt{2},...,i(M_{m}^\dagger - M_{m})/\sqrt{2}\}$
satisfy the same commutation relations as the quadrature operators  $\{\hat{x}_1,...,\hat{x}_m\}$ and $\{\hat{p}_{k+1},...,\hat{p}_m\}$
  respectively, and by Ref.~\cite{de2006symplectic} they can be extended to a ``symplectic basis,'' which is defined as a basis that is related to the original quadrature operators by a symplectic transform that corresponds to a Gaussian unitary $U$.
As a result, the state $\ket{\psi} := U(\ket{\hat{x}=0}_1\otimes \dots \otimes \ket{\hat{x} = 0}_k\otimes \ket{0}_{k+1} \otimes \dots \otimes \ket{0}_{m})$, is annihilated by all the $\{M_i\}_{i=1}^m$, meaning $M_i\ket{\psi} = 0$ for all $i$.
If an arbitrary multi-mode displacement $\boldsymbol{\alpha}$ is applied, the state $\mathcal{D}(\boldsymbol\alpha) \ket{\psi}$ is also a simultaneous eigenstate of the $\{M_i\}_{i=1}^m$ with eigenvalues determined by the displacement.
In this paper, when we say that a Gaussian measurement is characterized by
$\{M_{i}\}_{i=1}^{m}$, we mean that the $M_{i}$ satisfy the properties in List~\ref{list:gaussian_meas_props} and that the POVM elements of the measurement consist of the joint eigenstates of $\{M_i\}_{i=1}^m$, which are of the form
$\mathcal{D}(\boldsymbol\alpha)\ketbra{\psi}\mathcal{D}^\dagger(\boldsymbol\alpha)$
for all possible displacements $\boldsymbol\alpha$.
An important property of operators $M$ in the span of $\{M_{i}\}$ is that they have non-negative commutator with their adjoints, that is $[M, M^{\dagger}] \geq 0$.
Furthermore, any operator $M$ satisfying $[M, M^{\dagger}] \geq 0$ can be extended to a list of operators $\{M_i\}_{i=1}^m$ that characterize a Gaussian measurement such that $M$ is in the span of the $\{M_i\}_{i=1}^m$.
This can be verified in the case that $[M, M^{\dagger}] = 0$ by first observing that $M+M^\dagger$ and $iM^\dagger - iM$ are commuting self-adjoint operators so their span, which contains $M$, admits a basis of commuting self-adjoint operators.
This basis can be extended to a symplectic basis and in particular can be extended to a list of $m$ linearly independent mutually commuting self-adjoint operators~\cite{de2006symplectic}.
Similarly, if $[M, M^{\dagger}] > 0$ then it can be rescaled so that $[M, M^{\dagger}] = 1$ and the operators $\{(M + M^\dagger)/\sqrt{2}, i(M^\dagger - M)/\sqrt{2}\}$ then satisfy the commutation relations of
  $\hat{x}_m, \hat{p}_m$.
These can be extended to a symplectic basis
  $\{\hat{x}_1,\dots,\hat{x}_m, \hat{p}_1,\dots,\hat{p}_m\}$, and the list of operators
  $\{\hat{x}_1,\dots,\hat{x}_{m-1},M\}$ then satisfies the properties in List~\ref{list:gaussian_meas_props} for $k=m-1$.

As part of our methods of proof we analyze Gaussian measurements on a system $\mathsf{AC}$ when the subsystem $\mathsf{C}$ is initialized in vacuum.
If $\ketbra{\psi}_{\A\C}$ is a POVM element of the full $\mathsf{AC}$ measurement, then the effective POVM element on system $\mathsf{A}$ is
\begin{equation}
\Pi_\A = \text{tr}_\C\big[\text{id}_\A \otimes \ketbra{0}_\C \cdot \ketbra{\psi}_{\A\C}\big].
\label{eq:effectivepovm}
\end{equation}
The right-hand-side is mathematically equivalent to the calculation of the state on $\mathsf{A}$ after the system $\mathsf{AC}$ is initialized in the Gaussian state $\ketbra{\psi}_{\A\C}$ and then system $\mathsf{C}$ is projected into vacuum $\ketbra{0}_\C$. In such a situation, the resulting state on system $\mathsf{A}$ is a Gaussian pure state, which implies that the effective POVM element $\Pi_\A$ is a projector (possibly improper) onto a pure Gaussian state on system $\mathsf{A}$.

When error channels are applied before a joint measurement the
effective measurement may be separable. We use the definition that a
Gaussian measurement is separable if it has POVM elements that are all
convex combinations of positive product operators.
If a measurement is separable according to this definition then it cannot be used to generate entanglement during entanglement swapping nor to teleport a state such that entanglement with another system is preserved.
For completeness we elaborate on this point in Section~\ref{section_sep_meas}.

In this paper, the Gaussian measurements that we directly prove are separable can be expressed by circuits of the form shown in the left side of Diagram~\ref{dia:trace_to_meas}, which shows a Gaussian unitary $U_G$ acting on two modes $\mathsf{A,B}$ and two vacuum ancilla modes $\mathsf{E,F}$, followed by a noiseless Gaussian measurement $G$ of $\mathsf{A,B}$ and loss of the modes $\mathsf{E,F}$ that is mathematically represented by partial trace.
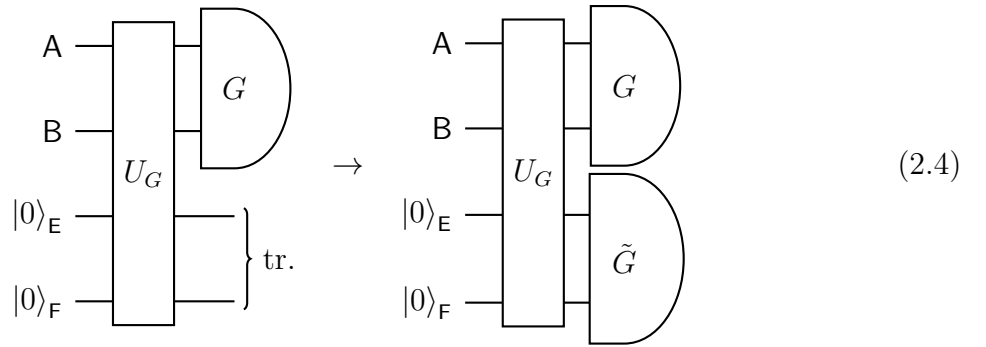
\begin{equation}
\label{dia:trace_to_meas}
\begin{quantikz}
 \lstick{$\mathsf{A}$} &\gate[wires=4]{U_G} &\gate[wires=2,style={xscale=.7,rounded rectangle,rounded rectangle left arc=none,inner sep=6pt,fill=white}]{G} \\
  \lstick{$\mathsf{B}$} &\ghost{U_G} & \ghost{G} \\
 \lstick{$\ket{0}_\mathsf{E}$} & \ghost{U_G} & \qw\rstick[wires=2]{tr.} \\
 \lstick{$\ket{0}_\mathsf{F}$} & \ghost{U_G}& \qw&
 \end{quantikz}
 \;\; \rightarrow \;\;
 \begin{quantikz}
  \lstick{$\mathsf{A}$} &\gate[wires=4]{U_G} &\gate[wires=2,style={xscale=.7,rounded rectangle,rounded rectangle left arc=none,inner sep=6pt,fill=white}]{G} \\
 \lstick{$\mathsf{B}$} &\ghost{U_G}& \ghost{G} \\
 \lstick{$\ket{0}_\mathsf{E}$} & \ghost{U_G} & \gate[wires=2,style={xscale=.7,rounded rectangle,rounded rectangle left arc=none,inner sep=6pt,fill=white}]{\tilde{G}} \\
 \lstick{$\ket{0}_\mathsf{F}$} & \ghost{U_G}&\ghost{\tilde{G}}
 \end{quantikz}
\end{equation}
To obtain a sufficient condition for the separability of certain Gaussian measurements of this form, we introduce a hypothetical Gaussian measurement $\tilde{G}$ on modes $\mathsf{E,F}$, as shown on the right side of Diagram~\ref{dia:trace_to_meas}.
The original measurement on the left side is equivalent to the measurement shown on the right side followed by loss of the outcome of the measurement $\tilde{G}$.
The measurement on the right side can be interpreted as a noiseless Gaussian measurement of the full system $\mathsf{A,B,E,F}$ where the modes $\mathsf{E,F}$ are initialized in vacuum.
This measurement is equivalent to some Gaussian measurement of just the modes $\mathsf{A,B}$, and we refer to its POVM elements as the effective system $\A\B$ POVM elements. If the effective system $\A\B$ POVM elements are all product operators then the original measurement on the left side of Diagram~\ref{dia:trace_to_meas} must have POVM elements that are all separable.
This follows because the effective system $\mathsf{AB}$ POVM elements of the measurement on the left side must be convex combinations of the system $\mathsf{AB}$ POVM elements of the measurement on the right, where the convex combinations result from the loss of the outcome of the measurement $\tilde{G}$.

\section{Loss condition for separability}
\label{section:loss_threshold}
In this section we establish that the effective measurement shown in the circuit
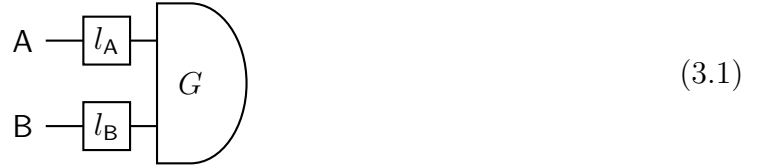
\begin{equation}
\begin{quantikz}
\lstick{$\mathsf{A}$} & \gate{l_\A} &\gate[wires=2,style={xscale=.7,rounded rectangle,rounded rectangle left arc=none,inner sep=6pt,fill=white}]{G} \\
\lstick{$\mathsf{B}$} & \gate{l_\B}& \ghost{G}
\end{quantikz}
\end{equation}
is separable for all Gaussian measurements $G$ when the loss-channel parameters $l_\A, l_\B$ satisfy
\begin{equation}
\label{eq:loss_threshold}
l_\A + l_\B \geq 1.
\end{equation}
The separability condition Eq.~\ref{eq:loss_threshold} is
proven in Prop.~\ref{prop:loss_threshold}.
The overall strategy is to first dilate the loss channels by applying beamsplitters with two vacuum ancilla modes $\mathsf{E}, \mathsf{F}$~\cite{Holevo2007,Ivan_2011}, as in the left side of Diagram~\ref{diagramlossdilation}.
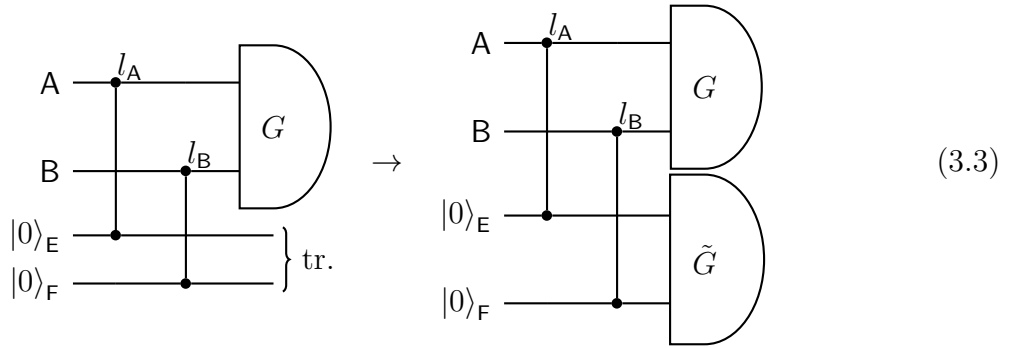
\begin{equation}
\begin{quantikz}
\lstick{$\mathsf{A}$} &\phase{l_\A}&\qw &\gate[wires=2,style={xscale=.7,rounded rectangle,rounded rectangle left arc=none,inner sep=6pt,fill=white}]{G} \\
\lstick{$\mathsf{B}$} &\qw&\phase{l_\B} & \ghost{G} \\
\lstick{$\ket{0}_\mathsf{E}$} & \ctrl{-2} &\qw & \qw\rstick[wires=2]{tr.} \\
\lstick{$\ket{0}_\mathsf{F}$} & \qw &\ctrl{-2}& \qw&
\end{quantikz}
\;\; \rightarrow \;\;
\begin{quantikz}
\lstick{$\mathsf{A}$} &\phase{l_\A}&\qw &\gate[wires=2,style={xscale=.7,rounded rectangle,rounded rectangle left arc=none,inner sep=6pt,fill=white}]{G} \\
\lstick{$\mathsf{B}$} &\qw&\phase{l_\B} & \ghost{G} \\
\lstick{$\ket{0}_\mathsf{E}$} & \ctrl{-2} &\qw & \gate[wires=2,style={xscale=.7,rounded rectangle,rounded rectangle left arc=none,inner sep=6pt,fill=white}]{\tilde{G}} \\
\lstick{$\ket{0}_\mathsf{F}$} & \qw &\ctrl{-2}&\ghost{\tilde{G}}
\end{quantikz}\label{diagramlossdilation}
\end{equation}
Here, vertical lines connecting
  two horizontal mode lines indicate beamsplitters. The beamsplitter
  labels indicate the beamsplitter reflectivity for the labeled mode,
  which is equal to the loss parameter. Then, consider replacing the trace on the two ancilla modes $\mathsf{E}$ and $\mathsf{F}$ by a Gaussian measurement
$\tilde{G}$ to be determined, as in the right side of Diagram~\ref{diagramlossdilation}.
 We show that
whenever $l_\A + l_\B \geq 1$, for all Gaussian measurements $G$ there
exists a corresponding Gaussian measurement $\tilde{G}$ for which the effective system $\A\B$ POVM elements of the combined $G,\tilde{G}$ measurement are all product operators.
As discussed after Diagram~\ref{dia:trace_to_meas}, it follows in this case that the
effective measurement on the left-hand side of Diagram~\ref{diagramlossdilation}
is separable. To this end, we first give an algebraic sufficient condition
for noiseless Gaussian measurements with vacuum ancillas to have POVM elements that are all product operators.

\begin{prop}
\label{prop:simplified_swapping}
Let $\mathsf{A,B,E,F}$ be one-mode bosonic systems where modes $\mathsf{E}$ and $\mathsf{F}$ are prepared in vacuum. Consider a noiseless Gaussian measurement $G$ of $\mathsf{ABEF}$ characterized by four operators $\{M_1, M_2, M_3, M_4\}$ as described after List~\ref{list:gaussian_meas_props}.
If there exists a complex linear combination of the $\{M_i\}$ of the form
\begin{equation}
\hat{M} = \hat{A} + \xi\hat{e}^\dagger + \zeta\hat{f}^\dagger,
\end{equation}
with $\hat{A}$ a non-zero linear combination of $\{\hat{a},  \hat{a}^\dagger\}$ on system $\mathsf{A}$
and $\xi,\zeta$ arbitrary complex coefficients, then the effective measurement on system $\mathsf{AB}$ has POVM elements that are all projectors onto pure product Gaussian states.
\end{prop}

\begin{proof}
Let $\ket{\phi}_{\mathsf{ABEF}}$ be a simultaneous eigenstate of the $M_i$.
The displacements of $\ket{\phi}_{\mathsf{ABEF}}$ are also simultaneous eigenstates,
and there is a displacement for which the eigenvalue of $\hat{M}$
is not zero.
Without loss of generality, assume that $\hat{M}\ket{\phi} = \lambda\ket{\phi}$ with
$\lambda \ne 0$.
By rescaling $\hat{M}$ we can assume that $\lambda=1$.
According to Eq.~\ref{eq:effectivepovm}, the effective system $\A\B$ POVM elements of the measurement $G$ can be expressed as
\begin{equation}
\Pi_{\mathsf{AB}} = \text{tr}_{\mathsf{EF}} \left(\text{id}_{\mathsf{AB}} \otimes \ketbra{0}_{\mathsf{EF}} \cdot \ketbra{\phi}_{\mathsf{AB}\mathsf{EF}}\right)
\end{equation}
or are related to $\Pi_{\A\B}$ by displacements.
Insert $\hat{M}$ in the partial trace to obtain
\begin{equation}
\Pi_{\mathsf{AB}} = \text{tr}_{\mathsf{EF}} \left(\text{id}_{\mathsf{AB}} \otimes \ketbra{0}_{\mathsf{EF}} \hat{M}\cdot \ketbra{\phi}_{\mathsf{AB}\mathsf{EF}}\right).
\end{equation}
The creation operators $\hat{e}^\dagger$ and $\hat{f}^\dagger$ annihilate vacuum from the right, so this simplifies to
\begin{equation}
\Pi_{\mathsf{AB}} = \text{tr}_{\mathsf{EF}} \left[\text{id}_{\mathsf{AB}} \otimes \ketbra{0}_{\mathsf{EF}} \hat{A}\cdot \ketbra{\phi}_{\mathsf{AB}\mathsf{EF}}\right] = \hat{A}\Pi_{\mathsf{AB}} .
\end{equation}
By taking adjoints on both sides, we also have $\Pi_{\mathsf{AB}} = \Pi_{\mathsf{AB}}\hat{A}^\dagger$.
By construction there is a one-mode Gaussian measurement that is characterized
by $\hat A$, and the $1$-eigenstate $\ket{\varphi}_{\mathsf{A}}$  of $\hat A$ is unique.
Intuitively, because $\hat A-\id_\A$ annihilates $\Pi_{\mathsf{AB}}$ on the left and $\hat A^\dagger - \id_\A$ annihilates it
on the right, $\Pi_{\mathsf{AB}}$ must be of the form $\Pi_{\mathsf{B}}\otimes \ketbra{\varphi}_{\mathsf{A}}$.
To verify this statement, let $\{\ket{k}_{\mathsf{B}}\}_{k=0}^\infty$  be an orthonormal basis
of $\mathsf{B}$, for example the number state basis. Then
$\id_\B = \sum_k \ketbra{k}_{\mathsf{B}}$ and we can express
\begin{align}
  \Pi_{\mathsf{AB}} &= \sum_k \ketbra{k}_{\mathsf{B}} \Pi_{\mathsf{AB}}\sum_{k'} \ketbra{k'}_{\mathsf{B}} \nonumber\\
         &= \sum_{k,k'} \ketbra{k}{k'}_{\mathsf{B}} \bra{k}_{\mathsf{B}}\Pi_{\mathsf{AB}} \ket{k^\prime}_{\mathsf{B}}.  \label{eq:prod1}
\end{align}
The operators $\Pi_{k,k'} := \bra{k}_{\mathsf{B}}\Pi_{\mathsf{AB}} \ket{k'}_{\mathsf{B}}$
on $\mathsf{A}$ are annihilated on the left by $\hat{A}-1$ and on the right by $\hat{A}^\dagger-1$.
They are therefore proportional to $\ketbra{\varphi}_{\mathsf{A}}$. Substituting accordingly
on the right-hand side of Eq.~\ref{eq:prod1} expresses   $\Pi_{\mathsf{AB}}$
as a tensor product of a projector acting on $\mathsf{B}$ with a projector onto a pure Gaussian state acting on $\mathsf{A}$.
Since $\Pi_{\mathsf{AB}}$ is a pure Gaussian state, the projector acting on $\mathsf{B}$
is also pure Gaussian.
Because all POVM elements of a Gaussian measurement are related by displacements, it follows that all effective system $\A\B$ POVM elements of the measurement $G$ are product operators.
\end{proof}

\begin{prop}
\label{prop:loss_threshold}
Let $G$ be a Gaussian measurement of modes $\mathsf{AB}$.
Consider the effective Gaussian measurement described in the circuit
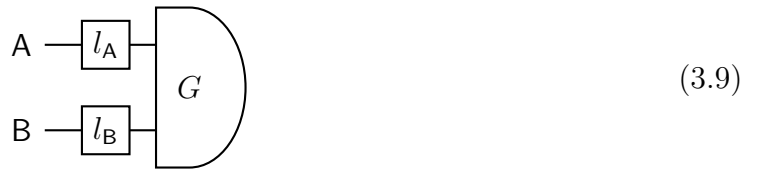
\begin{equation}
\label{dia:loss_threshold}
\begin{quantikz}
\lstick{$\mathsf{A}$} & \gate{l_\A} &\gate[wires=2,style={xscale=.7,rounded rectangle,rounded rectangle left arc=none,inner sep=6pt,fill=white}]{G} \\
\lstick{$\mathsf{B}$} & \gate{l_\B}& \ghost{G}
\end{quantikz}
\end{equation}
If the loss channel parameters $l_\A$ and $l_\B$ satisfy $l_\A + l_\B \geq 1$, then this
effective Gaussian measurement
is separable.
\end{prop}

\begin{proof}
If the proposition holds for all noiseless Gaussian measurements, then it holds for all, potentially noisy, Gaussian measurements. We therefore restrict $G$ to be a noiseless Gaussian measurement for the remainder of the proof.

We first dilate the loss in the circuit of the proposition to obtain
the circuit shown on the left of Diagram~\ref{diagramlossdilation}, then determine a
measurement $\tilde G$ as on the right of the diagram to prove separability.
We write $\tilde{a},\tilde{b}, \tilde{e}, \tilde{f}$ for the Heisenberg-evolved mode operators after the unitary implementing
the dilated loss.
In terms of the original mode operators, we have
\begin{align}
\tilde{a} &= \sqrt{1-l_{\mathsf{A}}} \hat a + \sqrt{l_{\mathsf{A}}} \hat e, \nonumber\\
\tilde{e} &= \sqrt{1-l_{\mathsf{A}}} \hat e - \sqrt{l_{\mathsf{A}}} \hat a, \nonumber\\
\tilde{b} &= \sqrt{1-l_{\mathsf{B}}} \hat b + \sqrt{l_{\mathsf{B}}} \hat f, \nonumber\\
\tilde{f} &= \sqrt{1-l_{\mathsf{B}}} \hat f - \sqrt{l_{\mathsf{B}}} \hat b.
\end{align}
If either $l_{\mathsf{A}}=1$ or $l_{\mathsf{B}}=1$, then one of the modes
is replaced by vacuum before the measurement, which implies separability of the effective measurement. For the remainder of the proof, we assume that
neither equality holds.

Choose $M_{1}$ and $M_{2}$ acting on modes $\mathsf{AB}$ so that they characterize $G$, as described after List~\ref{list:gaussian_meas_props}. The
  operators $M_{i}$ are independent commuting linear
  combinations of the four operators $\tilde{a},\tilde{b}, \tilde{a}^\dagger, \tilde{b}^\dagger$.
By independence, there is a non-zero linear combination $M$ of the $M_{i}$ that is in the
three dimensional span of $\tilde{a}$, $\tilde{a}^\dagger$ and $\tilde{b}^{\dagger}$. Write $M=\alpha \tilde{a}+\alpha'\tilde{a}^{\dagger}+\beta'\tilde{b}^{\dagger}$ for complex coefficients $\alpha, \alpha^\prime, \beta^\prime$.
In terms of the incoming mode operators,
\begin{equation}
M = \alpha\left(\sqrt{1-l_{\mathsf{A}}}  \hat{a} +  \sqrt{l_{\mathsf{A}}}  \hat{e}\right) + \alpha'\left(\sqrt{1-l_{\mathsf{A}}}  \hat{a}^\dagger  +  \sqrt{l_{\mathsf{A}}}  \hat{e}^\dagger\right) + \beta'\left(\sqrt{1-l_{\mathsf{B}}}  \hat{b}^\dagger  +  \sqrt{l_{\mathsf{B}}}\hat{f}^\dagger\right).
\end{equation}
We wish to apply Prop.~\ref{prop:simplified_swapping}, which requires
eliminating the terms of $M$ involving $\hat{e}$ and $\hat{b}^{\dagger}$ without
introducing terms involving $\hat{b}$ or $\hat{f}$.
Suppose the Gaussian measurement $\tilde G$ is applied to modes $\mathsf{EF}$, and suppose $\tilde M$ is one of a pair of operators that characterize $\tilde G$.  Then $M + \tilde M$ is a linear combination of the operators that characterize the full measurement on $\mathsf{ABEF}$, and it can be used with Prop.~\ref{prop:simplified_swapping}.
With foresight, let $\tilde M := \gamma \tilde{e} + \delta \tilde{f}^{\dagger}$ with $\gamma = -\alpha\sqrt{l_{\mathsf{A}}}/\sqrt{1-l_{\mathsf{A}}}$
and $\delta = \beta'\sqrt{1-l_{\mathsf{B}}}/\sqrt{l_{\mathsf{B}}}$.
The coefficient $\gamma$ of $\tilde{e}$ is chosen to cancel the
term involving $\hat{e}$ in $M$ and the coefficient $\delta$ of $\tilde{f}^{\dagger}$
is chosen to cancel the term involving $\hat{b}^{\dagger}$.
By design, no
unwanted terms are introduced.
If this assignment of $\tilde M $ is possible then $M+\tilde M$ can be expressed as $M+\tilde M = A + \xi \hat e^{\dagger}+ \zeta \hat f^{\dagger}$ for some complex $\xi,\zeta$ with $A$ acting on mode $\mathsf{A}$, in which case Prop.~\ref{prop:simplified_swapping} implies that the full measurement is separable across $\mathsf{A}$ and $\mathsf{B}$.
We next show that if $l_{\mathsf{A}}+l_{\mathsf{B}} \geq 1$, then there exists a Gaussian measurement $\tilde G$ such that $\tilde M := \gamma \tilde{e} + \delta \tilde{f}^{\dagger}$ is one of a pair that characterize $\tilde G$.
In fact, as discussed after List~\ref{list:gaussian_meas_props}, it suffices for $\tilde M$ to satisfy
$[\tilde M, \tilde M^{\dagger}] =|\gamma|^{2}- |\delta|^{2}\geq 0$.
In terms of
the loss parameters, this condition is expressed as
\begin{equation}
|\alpha|^2\left(\frac{l_\A}{1-l_\A}\right) - |\beta'|^2\left(\frac{1-l_\B}{l_\B}\right) \geq 0.
\label{ineq:lalb}
\end{equation}
As noted in Section~\ref{section:preliminaries}, because $M$ is an operator
in the span of a family characterizing a Gaussian measurement,
we have $[M,M^{\dagger}] \geq 0$. In terms of the coefficients in the expression
for $M$, this becomes $|\alpha|^{2}-|\alpha'|^{2}-|\beta'|^{2} \geq 0$. In particular the inequality $|\alpha|^{2} \geq |\beta'|^{2}$ is satisfied. Thus Eq.~\ref{ineq:lalb}
is satisfied if $l_{\mathsf{A}}/(1-l_{\mathsf{A}}) \geq  (1-l_{\mathsf{B}})/l_{\mathsf{B}}$.
Multiplying out the nonnegative denominators gives the inequality
$l_{\mathsf{A}}l_{\mathsf{B}} \geq (1-l_{\mathsf{A}})(1-l_{\mathsf{B}})$,
which can be rewritten as $l_{\mathsf{A}}+l_{\mathsf{B}} \geq 1$.
\end{proof}

\section{Extension to other phase-insensitive channels}
\label{section:sm_channels}
We now use Prop.~\ref{prop:loss_threshold} to determine conditions for separability of all Gaussian measurements when other phase-insensitive single-mode Gaussian error channels are applied prior to measurement. We find it convenient to begin with channels that take the form of amplification followed by loss.

\begin{prop}
\label{prop:amp_loss}
Let $G$ be a Gaussian measurement of modes $\mathsf{AB}$, let $\pamp_\A$ and $\pamp_\B$ be parameters of amplification channels, and let $l_\A$ and $l_\B$ be parameters of loss channels. The effective measurement shown in Diagram~\ref{dia:amp_loss}
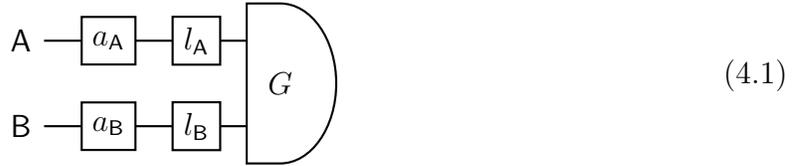
\begin{equation}
\label{dia:amp_loss}
\begin{quantikz}
\lstick{$\mathsf{A}$} & \gate{\pamp_\A} & \gate{l_\A} & \gate[wires=2,style={xscale=.7,rounded rectangle,rounded rectangle left arc=none,inner sep=6pt,fill=white}]{G} \\
\lstick{$\mathsf{B}$} & \gate{\pamp_\B} & \gate{l_\B} & \ghost{G}
\end{quantikz}
\end{equation}
is separable for all $G$ if and only if $l_\A + l_\B \geq 1$.
\end{prop}
\begin{proof}
If $l_\A + l_\B \geq 1$, then Prop.~\ref{prop:loss_threshold} applies, and the effective measurement is separable regardless of the amplification parameters $\pamp_\A$ and $\pamp_\B$. If $l_\A + l_\B < 1$, then simulating an all-Gaussian entanglement swapping scenario with a \cvbell{} measurement, similar to the one in Ref.~\cite{Hoelscher_Obermaier_2011}, and checking for entanglement~\cite{Duan_2000,Simon} of the two output modes demonstrates that entanglement swapping is possible regardless of the values of $\pamp_\A$ and $\pamp_\B$. This calculation is described in further detail in Prop.~\ref{prop:amp_loss_tight}. As a result, the \cvbell{} measurement cannot be separable if $l_\A + l_\B < 1$, so the inequality in Prop.~\ref{prop:loss_threshold} is tight even when arbitrary amplification channels act prior to loss.
\end{proof}

We next turn to the problem of establishing a condition for separability when the error channels take the form of loss $l_\A,l_\B$ followed by added noise with parameters $\pnoise_\A,\pnoise_\B$. We obtain this condition, given in Prop.~\ref{prop:loss_noise}, by re-parameterizing the channels as amplification followed by loss and applying Prop.~\ref{prop:amp_loss}. The relationship between the two parameterizations is given by Prop.~\ref{prop:amp_loss_equivalence}.

\begin{prop}
\label{prop:amp_loss_equivalence}
A non-entanglement-breaking Gaussian error channel that consists of loss $l$ followed by added noise $\pnoise$ is equivalent to a channel that consists of amplification with parameter $\pamp^\prime = \frac{1-l}{1-l-\pnoise}$ followed by a loss channel with parameter $l^\prime = l + \pnoise$ as shown in Diagram~\ref{dia:amp_loss_equivalence}.
\end{prop}

\begin{equation}
\label{dia:amp_loss_equivalence}
\begin{quantikz}
\lstick{$\rho$} & \gate{l} & \gate{\pnoise} & \qw
\end{quantikz}
=
\begin{quantikz}
\lstick{$\rho$} & \gate{\pamp^\prime} & \gate{l^\prime} & \qw
\end{quantikz}
\end{equation}

\begin{proof}
By applying the definitions in Section~\ref{section:preliminaries}, loss with parameter $l$ followed by noise with parameter $\pnoise$ acts on covariance matrices $V$ according to
\begin{equation}
V \mapsto (1-l)V + (l/2 + \pnoise)\cdot\mathbb{1}.
\end{equation}
 This can be matched to Eq.~\ref{eq:channel_cov_update} to recover the effective amplification and loss parameters $\pamp^\prime, l^\prime$. Matching terms proportional to $V$ requires $(1-l^\prime)\pamp^\prime = (1-l)$, and then matching the constant terms and solving for $l^\prime$ gives $l^\prime = l + \pnoise$. The denominator of $\pamp^\prime$ is positive as long as $l + \pnoise < 1$, which is equivalent to requiring that the channel is not entanglement-breaking~\cite{holevo2008entanglementbreaking}.
\end{proof}

\begin{prop}
\label{prop:loss_noise}
Let $G$ be a Gaussian measurement of modes $\A\B$, let $l_\A, l_\B$ be parameters of loss channels, and let $\pnoise_\A, \pnoise_\B$ be parameters of noise channels. The effective measurement shown in Diagram~\ref{dia:loss_noise}

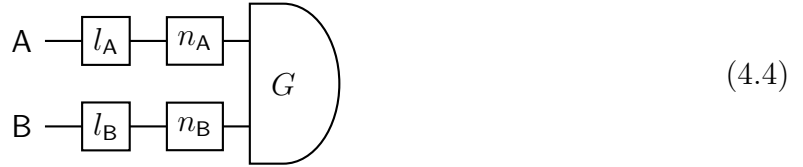
\begin{equation}
\label{dia:loss_noise}
\begin{quantikz}
\lstick{$\mathsf{A}$} & \gate{l_\A} & \gate{\pnoise_\A} & \gate[wires=2,style={xscale=.7,rounded rectangle,rounded rectangle left arc=none,inner sep=6pt,fill=white}]{G} \\
\lstick{$\mathsf{B}$} & \gate{l_\B} & \gate{\pnoise_\B} & \ghost{G}
\end{quantikz}
\end{equation}
is separable for all $G$ if and only if $l_\A + l_\B + \pnoise_\A + \pnoise_\B \geq 1$.
\end{prop}

\begin{proof}
By Prop.~\ref{prop:amp_loss_equivalence}, this is equivalent to the circuit
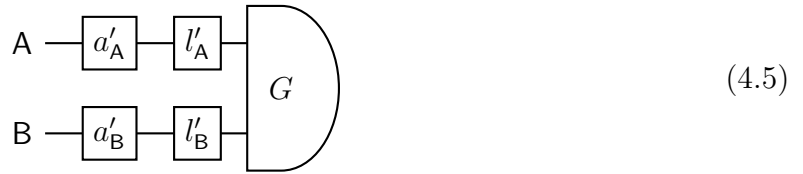
\begin{equation}
\begin{quantikz}
\lstick{$\mathsf{A}$} & \gate{\pamp^\prime_\A} & \gate{l^\prime_\A} & \gate[wires=2,style={xscale=.7,rounded rectangle,rounded rectangle left arc=none,inner sep=6pt,fill=white}]{G} \\
\lstick{$\mathsf{B}$} & \gate{\pamp^\prime_\B} & \gate{l^\prime_\B} & \ghost{G}
\end{quantikz}
\end{equation}
for loss parameters $l^\prime_\A = l_\A + \pnoise_\A$ and $l^\prime_\B = l_\B + \pnoise_\B$ and amplification parameters $\pamp^\prime_\A$ and $\pamp^\prime_\B$. Applying Prop.~\ref{prop:amp_loss} immediately gives the condition for separability $l_\A + l_\B + \pnoise_\A + \pnoise_\B \geq 1$.
\end{proof}

\section{All pairs of single-mode Gaussian channels}
\label{section:all_pairs}

In this section we show that Prop.~\ref{prop:amp_loss}, which covers error channels that are amplification followed by loss, can be extended to cover all pairs of single-mode Gaussian error channels.
To summarize the result, for any pair of single-mode Gaussian channels the question of whether all effective Gaussian measurements are separable can be answered with the following procedure. First, check if either channel in the pair is entanglement-breaking, which would immediately imply that the effective measurement is separable. Then, if any channel is unitarily equivalent to a channel that adds noise to a single quadrature, treat it as the identity channel.
Finally, reparameterize any remaining channels, up to unitary equivalence, as amplification followed by loss and apply Prop.~\ref{prop:amp_loss}.
By unitary equivalence we mean that two channels are equivalent if one can be transformed into the other by applying single-mode Gaussian unitaries before and after.
We emphasize that reparameterization up to unitary equivalence is suitable for our purpose because any single-mode unitaries after the channels can be absorbed into a general Gaussian measurement and any single-mode unitaries before the channels can be absorbed into state preparation. In either case this will not affect the determination of whether or not all effective Gaussian measurements are separable. We prove in Prop.~\ref{prop:single_mode_gaussian_channels} that reparameterization, up to unitary equivalance, of a non-entanglement-breaking single-mode Gaussian channel as amplification followed by loss is possible unless it is unitarily equivalent to a channel that adds noise to a single quadrature. We then show that a channel that adds noise to a single quadrature can be treated as identity, meaning $l = a = 0$, when the inequality in Prop.~\ref{prop:amp_loss} is evaluated.
Thus, in the remainder of this section we prove that the procedure outlined previously is sufficient to determine, for any pair of single-mode Gaussian channels, whether all effective Gaussian measurements are separable.

\begin{prop}
\label{prop:single_mode_gaussian_channels}
 All non-entanglement-breaking single-mode Gaussian channels are unitarily equivalent to amplification followed by loss, or unitarily equivalent to a channel that adds noise to only one quadrature.
\end{prop}
\begin{proof}
Single-mode Gaussian channels have been classified up to unitary equivalence in Ref.~\cite{Holevo2007}.
We use the notation from Refs.~\cite{Holevo2007,holevo2008entanglementbreaking,Ivan_2011}, which define four families of single-mode Gaussian channels denoted $A$, $B$, $C$ and $D$ and prove that any single-mode Gaussian channel is a member of one of these families.
Channels of type $A$ represent complete loss of one or both quadratures and are entanglement-breaking~\cite{Weedbrook_2012,holevo2008entanglementbreaking}. Channels of type $D$ are the complements of amplification channels and are also entanglement-breaking~\cite{holevo2008entanglementbreaking}.
Types $B$ and $C$ are further subdivided into $B_1$, $B_2$, $C_1$, and $C_2$ in Refs.~\cite{holevo2008entanglementbreaking,Ivan_2011}.
The channel $B_1$ adds a constant amount of noise to a single quadrature.
Matching the definitions in Ref.~\cite{Ivan_2011} to the definitions in Section~\ref{section:preliminaries} shows that channels of type $B_2$ are noise channels with parameter $\pnoise$, channels of type $C_1$ can be parameterized as loss $l$ followed by added noise $\pnoise_l$, and channels of type $C_2$ can be parameterized as amplification $\pamp$ followed by added noise $\pnoise_\pamp$.
These channels may or may not be entanglement-breaking depending on the specific values of the parameters.
 We showed that non-entanglement-breaking channels of type $C_1$ can be reparametrized as amplification followed by loss in Prop.~\ref{prop:amp_loss_equivalence}, and we show it for channels of type $B_2$ in Lemma~\ref{prop:b2_amp_loss} and for channels of type $C_2$ in Lemma~\ref{prop:c2_amp_loss}.
\begin{lem}
\label{prop:b2_amp_loss}
A non-entanglement-breaking noise channel with parameter $\pnoise$ can be parameterized as amplification with parameter $\pamp^\prime = \frac{1}{1-\pnoise}$ followed by loss with parameter $l^\prime = \pnoise$.
\end{lem}
\begin{proof}
A noise channel with parameter $\pnoise$ has the action on a covariance matrix $V$ of $V \mapsto V + \pnoise\cdot\mathbb{1}$ and is entanglement-breaking if and only if $\pnoise \geq 1$~\cite{holevo2008entanglementbreaking}. If $\pnoise < 1$, then matching to the action of amplification $\pamp^\prime$ followed by loss $l^\prime$ in Eq.~\ref{eq:channel_cov_update} shows that $\pamp^\prime = \frac{1}{1-\pnoise}$ and $l^\prime = \pnoise$.
\end{proof}

\begin{lem}
\label{prop:c2_amp_loss}
 A non-entanglement-breaking Gaussian error channel that consists of amplification $\pamp$ followed by added noise $\pnoise$ is equivalent to an amplification channel followed by a loss channel with parameters $\pamp^\prime = \frac{\pamp}{1-\pnoise}$ and $l^\prime = \pnoise$ respectively. This equivalence is shown in Diagram~\ref{dia:c2_amp_loss}.

\begin{equation}
\label{dia:c2_amp_loss}
\begin{quantikz}
\lstick{$\rho$} & \gate{\pamp} & \gate{\pnoise} & \qw
\end{quantikz}
=
\begin{quantikz}
\lstick{$\rho$} & \gate{\pamp^\prime} & \gate{l^\prime} & \qw
\end{quantikz}
\end{equation}
\end{lem}
\begin{proof}
A channel that consists of amplification with parameter $\pamp$ followed by noise with parameter $\pnoise$ is entanglement-breaking if and only if $\pnoise \geq 1$~\cite{holevo2008entanglementbreaking}. If $\pnoise < 1$, then applying the definitions in Section~\ref{section:preliminaries} shows that the channel acts on covariance matrices $V$ as
\begin{equation}
V \mapsto \pamp V + (\pamp - 1)/2\cdot\mathbb{1} + \pnoise\cdot\mathbb{1}.
\end{equation}
Matching this to the action of amplification $\pamp^\prime$ followed by loss $l^\prime$ in Eq.~\ref{eq:channel_cov_update} shows that $\pamp^\prime = \pamp/(1-\pnoise)$ and $l^\prime = \pnoise$.
\end{proof}

 As a result, all pairs of non-entanglement-breaking single-mode channels of types $B_2, C_1, C_2$ can be reparameterized as amplification followed by loss and are covered by Prop.~\ref{prop:amp_loss}. The only remaining non-entanglement-breaking channels
are of type $B_1$ and add noise to a single quadrature.
 This completes the proof of Prop.~\ref{prop:single_mode_gaussian_channels}.

 \end{proof}

To address channels of type $B_1$, which add noise to a single quadrature, we first observe that the amount of added noise can be made arbitrarily small by anti-squeezing that quadrature before the channel and squeezing it afterward.
As a result, these channels are arbitrarily close to the identity channel by unitary equivalence
and can be treated as such when combined with another channel on the second
system for the purpose of checking whether a pair of error channels makes all effective Gaussian measurements separable. To make this claim precise,
consider a channel $\cB$ that adds noise $\varepsilon$ to one quadrature. We can compose it with a channel $\cB'$ that adds $\varepsilon$ of
noise to the orthogonal quadrature to get the added noise channel $\cN = \cB'\circ\cB$ with
noise parameter $\varepsilon$. Following a pair of channels with another channel on one
of the two systems cannot
result in an inseparable effective measurement if the original pair makes all effective measurements separable. Since we can make $\varepsilon$ arbitrarily small
by the squeezing procedure mentioned above, for the purpose of checking whether
a pair of channels including a channel of type $B_{1}$ allows for inseparable effective
measurement, we can treat  the type $B_{1}$ channel as equivalent to the identity
channel.  On the other hand, if the other channel together with the identity channel
makes all effective measurements separable, then replacing the identity channel
with $\cB$ cannot change this fact. In conclusion, when combining a channel
of type $B_{1}$ with another channel, separability of all effective measurements
is equivalent to that when the channel of type $B_{1}$ is replaced by the identity.

This completes the proof that the procedure stated at the beginning of Section~\ref{section:all_pairs} is sufficient to determine, for any pair of single-mode Gaussian channels, whether all effective Gaussian measurements are separable.

\section{Proof of Prop.~\ref{prop:amp_loss} using dual channels}
\label{section:dual_channels}

An alternative method to obtain the inequality in Prop.~\ref{prop:amp_loss} is to consider dual channels and apply the characterization found in Ref.~\cite{Filippov_Ziman} of pairs of channels that annihilate entanglement of Gaussian states, meaning that the output of every Gaussian input state is separable.
If $\Phi$ is a channel expressed as a completely-positive map from an input state $\rho$ to an output state $\Phi(\rho)$, then the dual channel is the completely-positive map $\Phi^*$ from input bounded operators $\Pi$ to output bounded operators $\Phi^*(\Pi)$ such that Eq.~\ref{eq:dual_equivalence} holds for all states $\rho$ and all bounded operators $\Pi$.
\begin{equation}
  \label{eq:dual_equivalence}
\text{tr}\left[\Phi(\rho) \cdot \Pi \right] = \text{tr}\left[\rho \cdot \Phi^*(\Pi) \right]
\end{equation}
For further information about dual Gaussian channels, we refer to Ref.~\cite{Ivan_2011}. In particular we use the fact from Ref.~\cite{Ivan_2011} that the dual of a loss channel, when restricted to density operators, is proportional to an amplification channel and vice versa.
The proportionality constants do not affect our analyses of separability.

To determine whether a joint Gaussian measurement on modes $\A$ and $\B$
after error channels $\Phi_{\A}$ and $\Phi_{\B}$ results in a separable effective measurement, we can apply the dual channels $\Phi_{\A}^{*}$ and $\Phi_{\B}^{*}$
to the POVM elements of the joint measurement. These POVM elements are proportional to the density operators of Gaussian
states, so if $\Phi_{\A}^{*}\otimes \Phi_{\B}^{*}$ annihilates entanglement of Gaussian states then the effective measurement is
separable. It suffices to check POVM elements that are pure Gaussian states,
which can be prepared by applying a Gaussian unitary to vacuum states. Therefore, it suffices to consider circuits of the form

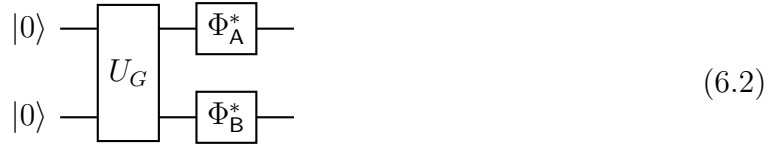
\begin{equation}
\begin{quantikz}
\lstick{$\ket{0}$} & \gate[wires=2]{U_G} & \gate{\Phi_\A^*} & \qw \\
\lstick{$\ket{0}$} & \ghost{U_G} & \gate{\Phi_\B^*} & \qw \\
\end{quantikz}
\end{equation}
where $U_G$ is an arbitrary Gaussian unitary. The dual of an amplification channel with parameter $\pamp$ is a loss channel with parameter $l^* = 1-\frac{1}{\pamp}$ and the dual of a loss channel with parameter $l$ is an amplification channel with parameter $\pamp^* = \frac{1}{1-l}$, as stated in Theorem 9 of Ref.~\cite{Ivan_2011}.
As a result, to obtain the inequality in Prop.~\ref{prop:amp_loss} it suffices to analyze the circuit in Diagram~\ref{dia:dual_sep}
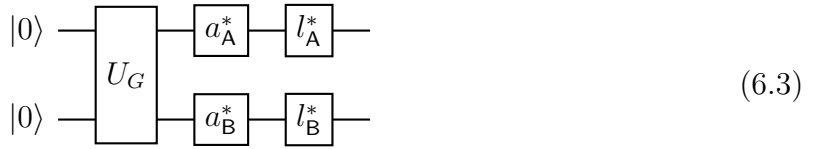
\begin{equation}
  \label{dia:dual_sep}
\begin{quantikz}
\lstick{$\ket{0}$} & \gate[wires=2]{U_G} & \gate{\pamp_\A^*} & \gate{l_\A^*} & \qw \\
\lstick{$\ket{0}$} & \ghost{U_G} & \gate{\pamp_\B^*} & \gate{l_\B^*} &\qw \\
\end{quantikz}
\end{equation}
for $\pamp^* = \frac{1}{1-l}$ and $l^* = 1-\frac{1}{\pamp}$ and determine the parameters at which an arbitrary Gaussian input state becomes separable.
Analysis of this scenario appears in Ref.~\cite{Filippov_Ziman}. The single-mode Gaussian channels considered in Ref.~\cite{Filippov_Ziman}
are parameterized as $\Phi(\kappa, \mu)$, where the action of $\Phi(\kappa,\mu)$ on covariance matrices is given by
\begin{equation}
V \mapsto \kappa V + \mu \cdot \mathbb{1}
\end{equation}
where $\kappa$ and $\mu$ are nonnegative real parameters.
This parameterization is general enough to cover any combination of loss, amplification, and added noise as defined in Section~\ref{section:preliminaries}. In terms of these parameters Ref.~\cite{Filippov_Ziman} proves the following condition for entanglement annihilation.

 \begin{prop}\label{FZfact}~\cite{Filippov_Ziman}
 The channel $\Phi(\kappa_{A} , \mu_{A} ) \otimes \Phi(\kappa_{B} , \mu_{B} )$ annihilates entanglement of all two-mode gaussian states if and only if $\kappa_{A} \mu_{B} + \kappa_{B} \mu_{A} \geq \tfrac{1}{2}(\kappa_{A} + \kappa_{B})$.
\end{prop}

If the channels take the form of amplification followed by loss, then this result can be conveniently written according to Prop~\ref{prop:amp_loss_direct}.
\begin{prop}
\label{prop:amp_loss_direct}
Consider amplification channels with parameters $\pamp^*_\A, \pamp^*_\B$ followed by loss channels with parameters $l^*_\A, l^*_\B$ that act independently on an arbitrary two-mode Gaussian pure state on $\A\B$, as shown in Diagram~\ref{dia:amp_loss_direct}.
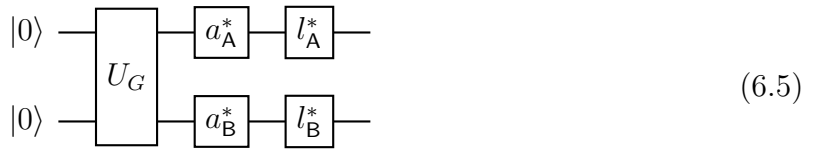
\begin{equation}
  \label{dia:amp_loss_direct}
\begin{quantikz}
\lstick{$\ket{0}$} & \gate[wires=2]{U_G} & \gate{\pamp^*_\A} & \gate{l^*_\A} & \qw \\
\lstick{$\ket{0}$} & \ghost{U_G} & \gate{\pamp^*_\B} & \gate{l^*_\B} &\qw \\
\end{quantikz}
\end{equation}
These channels annihilate entanglement of all two-mode Gaussian states if and only if
\begin{equation}
\label{eq:amp_loss_direct}
\frac{1}{\pamp^*_\A} + \frac{1}{\pamp^*_\B} \leq 1.
\end{equation}
\end{prop}
\begin{proof}
In the notation of Prop.~\ref{FZfact}, amplification $\pamp^*$ followed by loss $l^*$ is equivalent to $\Phi(\kappa,\mu)$ with $\kappa = \pamp^*(1-l^*)$ and $\mu = \kappa/2 + l^* - 1/2$. Substituting $\kappa_i, \mu_i$ for $i \in \{A,B\}$ into the inequality in Prop.~\ref{FZfact} leads to
\begin{equation}
\kappa_\A \kappa_\B \geq \kappa_\A(1-l^*_\B) + \kappa_\B(1-l^*_\A).
\end{equation}
and further substitution for $\kappa_\A,\kappa_\B$ and cancellation of $(1-l^*_\A)(1-l^*_\B)$ from both sides leads to
\begin{equation}
\pamp^*_\A \pamp^*_\B \geq \pamp^*_\A + \pamp^*_\B
\end{equation}
and the result follows.
\end{proof}

Now, substitution of the dual channel parameters $\pamp^* = \frac{1}{1-l}$ and $l^* = 1-\frac{1}{\pamp}$ into Eq.~\ref{eq:amp_loss_direct} gives the inequality
\begin{equation}
(1-l_\A) + (1-l_\B) \leq 1,
\end{equation}
which leads directly to the inequality $l_\A + l_\B \geq 1$, which is independent of
amplification parameters as asserted by Prop.~\ref{prop:amp_loss}.

\section{Conclusion}
In conclusion, we study two-mode Gaussian measurements that are made after independent single-mode Gaussian error channels and show that, if the error channels are amplification with parameters $\pamp_\A, \pamp_\B$ followed by loss with parameters $l_\A, l_\B$, then all effective Gaussian measurements are separable if and only if $l_\A + l_\B \geq 1$.
If the error channels are instead parameterized as loss $l_\A,l_\B$ followed by noise $\pnoise_\A, \pnoise_\B$, then this condition becomes $l_\A + l_\B + \pnoise_\A + \pnoise_\B \geq 1$.
The standard \cvbell{} measurement is a joint Gaussian measurement that remains inseparable for all error channels whose parameters do not satisfy this inequality.
Up to unitary equivalence, all pairs of non-trivial and non-entanglement-breaking single-mode Gaussian channels that act before the Gaussian measurement can be reduced to these cases.
In particular, for a given pair of independent single-mode Gaussian error channels this analysis is sufficient to determine whether or not entanglement swapping is possible with a joint Gaussian measurement, even when arbitrary input states are allowed.

Future work can extend this analysis to joint Gaussian measurements on $n>1$ modes of party $\mathsf{A}$ and $m>1$ modes of party $\mathsf{B}$, or to correlated Gaussian error channels, or to non-Gaussian bosonic error channels that are experimentally relevant.

\appendix

\section{The inequality in Prop.~\ref{prop:amp_loss} is tight}
\begin{prop}
\label{prop:amp_loss_tight}
For error channels consisting of amplification with parameters $a_\A,a_\B$ followed by loss with parameters $l_\A,l_\B$, all-Gaussian entanglement swapping is possible if $l_\A + l_\B < 1$ regardless of the amplification parameters.
\end{prop}
\begin{proof}
This follows by combining the arguments in Section~\ref{section:dual_channels} with the fact that entanglement swapping using two copies of the infinitely squeezed resource state, given by $\sum_i^\infty \ket{i}\otimes\ket{i}$ in the Fock basis, and conditioning on an outcome of the swapping measurement associated to an inseparable POVM element will necessarily result in an inseparable output state. However, we also provide in this section a direct covariance matrix calculation for all-Gaussian entanglement swapping with finitely squeezed resource states and a \cvbell{} measurement that takes place after amplification and loss channels are applied. We first compute the full four-mode covariance matrix for a pair of two-mode-squeezed
states subjected to loss channels with parameters $l_\A, l_\B$ and amplification channels with parameters $a_\A,a_\B$. Then we compute the resulting two-mode covariance matrix after a \cvbell{} measurement is applied to the modes that experience the amplification and loss. We use the formula from Ref.~\cite{Weedbrook_2012} to compute the post-measurement covariance matrix, which has the form

\begin{equation}
\label{eq:covmatstd}
\begin{pmatrix}
n_\A & 0 & c & 0\\
0 & n_\A & 0 & -c\\
c & 0 & n_\B & 0\\
0 & -c & 0 & n_\B
\end{pmatrix} \\
\end{equation}
for non-negative real parameters $n_\A,n_\B,c$. In terms of $a_\A,a_\B,l_\A,l_B$ and the two-mode-squeezing parameter $r$ we define for convenience $\kappa_\A := a_\A(1-l_\A), \kappa_\B := a_\B(1-l_\B)$ and $\eta := 1-l_\A-l_\B$ and find that

\begin{equation}
\label{eq:nc_subs}
\begin{aligned}
    n_\A &= \frac{2 \kappa_\A + \kappa_\B +2 (\kappa_A + \kappa_B - 2\eta)\cosh{2r} + \kappa_\B\cosh{4r}}{4 (\kappa_\A + \kappa_B - 2\eta + (\kappa_\A + \kappa_B)\cosh{2r})} \\
    n_\B &= \frac{ 2\kappa_\B + \kappa_\A +2 (\kappa_A + \kappa_B - 2\eta)\cosh{2r} + \kappa_\A\cosh{4r}}{4 (\kappa_\A + \kappa_B - 2\eta + (\kappa_\A + \kappa_B)\cosh{2r})} \\
    c &= \frac{2\sqrt{\kappa_\A \kappa_B}(\cosh{r}\sinh{r})^2}{\kappa_A + \kappa_B - 2\eta + (\kappa_A + \kappa_B)\cosh{2r}}
\end{aligned}
\end{equation}

According to Theorem 2 from Ref.~\cite{Duan_2000}, the state is entangled if and only if the inequality

\begin{equation}
\label{eq:duan}
    2 a^2 n_\A + 2\frac{n_\B}{a^2} - 4c - a^2 - \frac{1}{a^2} < 0
\end{equation}
is satisfied, where $a^2 := \sqrt{\frac{2n_\B-1}{2n_\A-1}}$. We note that the entries of the covariance matrix in Eq.~\ref{eq:covmatstd} are (co-)variances of quadratures, while Ref.~\cite{Duan_2000} uses the convention that the analogous covariance matrix's entries are twice the (co-)variances of quadratures.

Substituting Eqs.~\ref{eq:nc_subs} into Eq.~\ref{eq:duan} leads to
\begin{equation}
-2\frac{(1 - l_\A - l_\B)}{\sqrt{\kappa_\A \kappa_\B}} < 0
\end{equation}
in the limit that $r$ goes to infinity.
The parameters $\kappa_\A,\kappa_\B$ are nonnegative, so this inequality is satisfied when $l_\A + l_\B < 1$ regardless of $a_\A,a_\B$.

\end{proof}

\section{Separable measurements cannot swap or teleport entanglement}
{
\label{section_sep_meas}
\newcommand{\Ap}{\A^\prime}
\newcommand{\Bp}{\B^\prime}
\begin{prop}
  \label{prop_sep_meas}
Consider a joint system $\Ap,\A,\B,\Bp$ in an initial state $\rho_{\Ap\A\B\Bp}$ that is separable across the $\Ap\A$-$\Bp\B$ partition and consider a measurement of the $\A\B$ subsystem with an outcome associated to some POVM element $\Pi_{\A\B}$ that is a convex combination of positive product operators across the $\A$-$\B$ partition.
Assume for simplicity that the measured system $\A\B$ is lost after the measurement.
The state of the system $\Ap\Bp$ after the measurement is separable.
\end{prop}
\begin{proof}
The unnormalized state after the measurement, denoted $\sigma_{\Ap\Bp}$, is
\begin{equation}
  \label{eq_final_state}
\sigma_{\Ap\Bp} = \text{tr}_{\A\B}\left[\rho_{\Ap\A\B\Bp} \cdot \id_{\Ap\Bp}\otimes \Pi_{\A\B} \right].
\end{equation}
By assumption the POVM element $\Pi_{\A\B}$ can be decomposed as $\sum_i \lambda_i A_i \otimes B_i$ for positive coefficients $\lambda_i$ and positive operators $\A_i,\B_i$.
Furthermore, the initial state $\rho_{\Ap\A\B\Bp}$ is assumed separable and can be decomposed according to $\rho_{\Ap\A\B\Bp} = \sum_j \gamma_j \rho_{\Ap\A}\otimes\rho_{\B\Bp}$ for states $\rho_{\Ap\A}$ and $\rho_{\B\Bp}$ and positive coefficients $\gamma_j$.
Substituting into Eq.~\ref{eq_final_state} leads to
\begin{equation}
  \sigma_{\Ap\Bp} = \sum_{i,j} \lambda_i\gamma_j \left(\text{tr}_{\A} \left[\rho_{\Ap\A}\cdot \id_{\Ap}\otimes A\right]\right) \otimes \left( \text{tr}_{\B}\left[\rho_{\B\Bp}\cdot B\otimes\id_{\Bp} \right]\right)
\end{equation}
which is separable across the $\Ap$-$\Bp$ partition.
\end{proof}

Prop.~\ref{prop_sep_meas} immediately implies that entanglement swapping with a separable measurement on $\A\B$ cannot generate entanglement across the $\Ap$-$\Bp$ partition. Similarly, if the state of $\Ap\A$ is initially entangled and a separable measurement of $\A\B$ is used to attempt to teleport the state of system $\A$ to the system $\Bp$, then Prop.~\ref{prop_sep_meas} implies that the output state of $\Ap\Bp$ is unentangled and therefore entanglement with $\Ap$ has not been preserved during teleportation.
}

\begin{acknowledgments}
  The authors thank Arik Avagyan, Akshay Seshadri, and Victor Albert for helpful comments. A. Kwiatkowski, E.S., S.A., A. Kyle, and  C.R. acknowledge support from the Professional Research Experience Program (PREP) operated jointly by NIST and the University of Colorado. This work includes contributions of the National Institute of
  Standards and Technology, which are not subject to U.S. copyright.
\end{acknowledgments}


\bibliography{gaussian_swapping}

\end{document}